\documentclass[sigconf,authorversion,nonacm]{acmart}

\usepackage[utf8]{inputenc} 
\usepackage[T1]{fontenc}    
\usepackage{hyperref}       
\usepackage{url}            
\usepackage{booktabs}       
\usepackage{amsfonts}       
\usepackage{nicefrac}       
\usepackage{microtype}      
\usepackage{lipsum}
\usepackage{graphicx}
\usepackage{comment}
\usepackage{multirow}
\usepackage{subcaption}
\usepackage[linesnumbered,lined,ruled,vlined]{algorithm2e}

\usepackage{amsthm}
\newtheorem{definition}{Definition}
\newtheorem{proposition}{Proposition}
\newtheorem{lemma}{Lemma}

\AtBeginDocument{%
  \providecommand\BibTeX{{%
    \normalfont B\kern-0.5em{\scshape i\kern-0.25em b}\kern-0.8em\TeX}}}

\begin{document}
\pagestyle{plain}
\title{How Fair is Fairness-aware Representative Ranking and Methods for Fair Ranking}

\author{Akrati Saxena}

\affiliation{%
  \institution{Department of Mathematics and Computer Science\\
  Eindhoven University of Technology}
  \country{the Netherlands}
}
\email{a.saxena@tue.nl}

\author{George Fletcher}
\affiliation{%
  \institution{Department of Mathematics and Computer Science\\
  Eindhoven University of Technology}
  \country{the Netherlands}}
\email{g.h.l.fletcher@tue.nl}

\author{Mykola  Pechenizkiy}
\affiliation{%
  \institution{Department of Mathematics and Computer Science\\
  Eindhoven University of Technology}
  \country{the Netherlands}
}
\email{m.pechenizkiy@tue.nl}

\renewcommand{\shortauthors}{Saxena et al.}

\begin{abstract}
Rankings of people and items has been highly used in selection-making, match-making, and recommendation algorithms that have been deployed on ranging of platforms from employment websites to searching tools. The ranking position of a candidate affects the amount of opportunities received by the ranked candidate. It has been observed in several works that the ranking of candidates based on their score can be biased for candidates belonging to the minority community. In recent works, the fairness-aware representative ranking was proposed for computing fairness-aware re-ranking of results. The proposed algorithm achieves the desired distribution of top-ranked results with respect to one or more protected attributes. In this work, we highlight the bias in fairness-aware representative ranking for an individual as well as for a group if the group is sub-active on the platform. We define individual unfairness and group unfairness and propose methods to generate ideal individual and group fair representative ranking if the universal representation ratio is known or unknown. The simulation results show the quantified analysis of fairness in the proposed solutions. The paper is concluded with open challenges and further directions. 
\end{abstract}



\maketitle

\section{Introduction}\label{intro}



The ranking problem is encountered in many different applications, including ordering job candidates, providing results for search queries, providing suggestions and recommendations \cite{castillo2019fairness, cohen2019efficient, sanchez2020does}. There are several works that have highlighted the bias in placing the elements in top-$k$ returned results \cite{yang2017measuring, zehlike2017fa, kay2015unequal}. It is important to understand the limitations of these methods and update the models to achieve fair outcomes for sensitive scenarios as they have been applied to everyday applications \cite{corbett2017algorithmic,dwork2012fairness,friedler2016possibility}. There are several works on proposing fairness-aware ranking methods, including \cite{asudeh2019designing, biega2018equity, castillo2019fairness, celis2017ranking, zehlike2017fa, yang2017measuring, singh2018fairness}. The recent works have investigated both notions of fairness, (i) individual fairness \cite{dwork2012fairness, binns2020apparent, ilvento2019metric}, and (ii) group fairness \cite{pedreshi2008discrimination, lohia2019bias, pedreschi2009measuring}.

Geyik et al. \cite{geyik2019fairness} proposed fairness-aware representative ranking (FRR) maintains the ratio of candidates for each category in top-k elements based on their ratio in the eligible candidates so that each group will get a fair number of resources. They verified the proposed method by applying FRR on LinkedIn Talent Search and observed three-times improvement in the fairness without affecting the quality of the results. However, there is one big concern with the proposed method is what representation ratio for each group of candidates based on protected attributes should be maintained? We show that the deserving candidates of a group miss out on the opportunities by getting ranked lower in the returned list if the universal representation ratio of the candidates is not known. 


\textbf{Example.} Let's assume that for a search query $r$, the set of eligible candidates in Universe is represented by $U$ and the set of eligible candidates on a platform $\mathbb{L}$ is represented by $L$. For a protected attribute $A$ having two values the candidates are divided into two groups $G_1$ and $G_2$. The elements belonging to group $G_1$ are $\{b_1, b_2, b_3, b_4, b_5, b_6, b_7, b_8, b_9 , b_{10}\}$, and the ranking of candidates is $b_1 > b_2 > b_3 > b_4 > b_5 > ... > b_{10}$, therefore candidate $b_1$ has the highest ranking 1. The candidates belonging to group $G_2$ are $\{g_1, g_2, g_3, g_4, g_5\}$ and their ranking is $g_1 > g_2 > g_3 > g_4> g_5$. Now, we assume that all candidates of group $G_1$ joins the platform $\mathbb{L}$, and group $G_2$ is sub-active (or less active) on platform $L$, therefore, only two candidates $(g_1, g_2)$ join the platform $\mathbb{L}$ as shown in Figure 1. Both set $U$ and $L$ will apply fair representative ranking by maintaining the ratio of candidates for each protected group when top-$k$ candidates are requested; The fair representative ranking is explained in detail in Section 2.1. In this example, set $U$ will maintain the ratio $(G_1: G_2) = (10 : 5)$ and set $L$ will maintain the ratio $(G_1: G_2) = (10 : 2)$. \\
\textbf{If an employer requests top-6 candidates for search query $r$:}

\begin{itemize}
    \item $U$ will return: $(b_1 > g_1 > b_2 > b_3 > g_2 > b_4)$.
    \item $L$ will return: $(b_1 > b_2 > b_3 > g_1 > b_4 > b_5)$.
\end{itemize}

\begin{figure}[t]
    \centering
    \includegraphics[width=\linewidth]{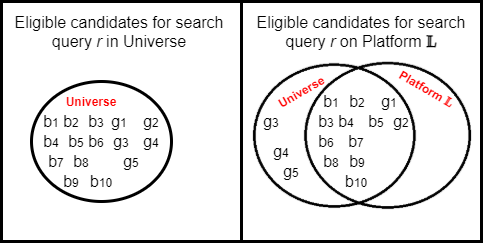} 
    \caption{Example for displaying the unfairness in representative ranking for a sub-active group.}
    \label{fig:my_label}
\end{figure}

You can observe that even though $g_2$ deserves the opportunity and is active on platform $\mathbb{L}$, the candidate has not been selected in top-6 when the FRR was applied on the platform $\mathbb{L}$. The main reason for this unfairness is that the \textit{representative ranking on platform $\mathbb{L}$ (LRR)} maintained the candidates' ratio for protected attributes based on $L$ and did not consider the universal representation ratio. Therefore, the eligible candidates miss the opportunity if other members belonging to that group either do not join or are sub-active on the platform.


This paper presents how the existing approach of fairness representative ranking is inadequate to generate a fair ranking given the protected attributes. We also present the definitions of unfairness that should be considered while proposing a fairness-aware ranking method. We discuss how the system should be extended to mitigate bias in ranking problems if the complete information is unknown. We further propose the updated definition of metrics in the case of missing data to mitigate the bias while evaluating a system's fairness. We discuss open challenges associated with the problem while applying fair representative ranking and further directions that might be considered for solving these challenges.

\section{Prerequisite}\label{prerequisite}

\subsection{Fair Representative Ranking}

The Fair Representative Ranking (FRR) method works as follows. 

\begin{enumerate}
    \item Partition the set of qualified candidates into different groups based on protected attribute values. For example, in the case of gender as a protected attribute, the candidates will be divided into two groups, males and females.
    
    \item Rank the candidates in each group according to their scores.
    
    \item To generate representative ranking, merge these groups in such a way that the proportion of each group in the ranked list for every index of the list is similar to their corresponding proportions in the set of qualified candidates. Note that the merging will preserve the ordering of the candidates that was followed within their group.
\end{enumerate}

\subsection{Terminologies}

\begin{itemize}
\item $r$: The search query (the set of skills required for a job).
\item $U=\{c_1, c_2, c_3, ....\}$: The universal set of eligible candidates for search query $r$.
\item $L$ $(L \subseteq U)$: The set of eligible candidates on platform $\mathbb{L}$.
\item $A=\{a_1, a_2, a_3, .... a_l\}$: The protected attribute (or a group of attributes) is represented by $A$ and it contains the value of protected attribute.
\item $A(c_i)$: The protected attribute value of candidate $c_i$.
\item For the given protected attributes $A$, the qualified candidates can be divided into $n$ groups $(G_1, G_2, ... G_n)$.
\item The fraction of eligible candidates for each protected value in $U$ is $\{p_{a_1}, p_{a_2}, ...,  p_{a_l}\}$. 
\item $\tau_{U,r, A}$ represents the representative ranking of candidates from set $U$ for the search query $r$, based on protected attributes $A$.
\item $\tau_{U,r, A}^k$ denotes top $k$ elements in the ranking.
\item $\tau_{U,r, A}(c_i)$ denotes the rank of candidate $c_i$, by the abuse of notation.
\item $\tau_{U,r, A}[i]$ denotes $i_{th}$ element of the ranking.
\item $\tau_{U,r, A}$ and $\tau_{L,r, A}$ are also referred as URR and LRR, respectively.
\end{itemize}

\section{Fairness Definitions}\label{problem} 

In this section, we discuss fairness and bias definitions from individual and group perspectives.

\subsection{Benefited Candidate}

If a candidate belongs to top-$k$ ranked candidates based on search query $r$, the candidate will be called the ``benefited candidate".

\begin{definition}\label{def1}
  For the given set $L$, skill set $r$, protected attribute $A$, and the selected candidates $k$, a candidate $c_i$ is benefited candidate if $c_i \in \tau^k_{L,r, A}$. 
\end{definition}

Note that while defining benefited candidate, the LRR is considered as the opportunity is being provided by platform $\mathbb{L}$.

\subsection{Individual Unfairness}

If a candidate has the universal representation ranking less than or equal to $k$, then the candidate should be benefited. However, if due to some reasons, on platform $\mathbb{L}$, the representative ranking of the candidate is greater than $k$, the candidate will not be benefited, and therefore it shows that the candidate has been treated unfairly by the platform $\mathbb{L}$. The example of Individual unfairness was discussed in the Introduction where $g_2$ was treated unfairly. 

\begin{definition}\label{def2}
  For the given set $U$ and $L$, search query $r$, and protected attribute $A$, if top-$k$ candidates are selected from $L$, the individual unfairness for a candidate $c_i$ is defined as, 
  
$IUF(c_i)=\left\{\begin{matrix} 1, & c_i \notin \tau^k_{L,r, A} and c_i \in \tau^k_{U,r, A} \\  0, & otherwise \end{matrix}\right.$
\end{definition}


\subsection{Favored Candidate}

If a candidate has universal representation ranking greater than $k$ then the candidate should not be benefited, but if on the platform $\mathbb{L}$, the representation ranking of the candidate is less than or equal to $k$, the candidate will be benefited, and therefore, this candidate will be referred to as `favored candidate'. In the example, candidate $b_5$ is the favored candidate as it would not have been benefited based on the universal representation ranking.

\begin{definition}\label{def3}
  For the given set $U$ and $L$, search query $r$, and protected attribute $A$, if top-$k$ candidates are selected from $L$, a candidate $c_i$ is favored candidate if, $c_i \in \tau^k_{L,r, A}$ and $c_i \notin \tau^k_{U,r, A}$.
\end{definition}

\subsection{Group Unfairness}

If the number of selected candidates from a group $G_i$ in LRR is less than the number of candidates that should have been selected based on URR, then the group $G_i$ has been treated unfairly on platform $\mathbb{L}$. For example, if 10 people from a group $G_i$ belong to top-$k$ people in set $U$ and less than 10 people from this group are selected from set L, then the representative ranking of set $L$ is called unfair for group $G_i$.

\begin{definition}\label{def5}
In top-$k$ selected candidates, group $G_i$ has been treated unfairly by platform $\mathbb{L}$, if  $|\{c_i | c_i \in G_i \; and \; c_i \in \tau^k_{L,r, A}\}| < |\{c_i | c_i \in G_i \; and \; c_i \in \tau^k_{U,r, A}\}|$.
\end{definition}

\subsection{Favored Group}

A group $G_i$ is favored if, in top-$k$ candidates, the number of benefited candidates from the group $G_i$ on platform $\mathbb{L}$ is more than the number of benefited candidates from the group $G_i$ on $U$. The favored group explains the opposite cases of group unfairness.

\begin{definition}\label{def7}
    In top-$k$ selected candidates, group $G_i$ is favored by platform $\mathbb{L}$, if  $|\{c_i | c_i \in G_i \; and \; c_i \in \tau^k_{L,r, A}\}| > |\{c_i | c_i \in G_i \; and \; c_i \in \tau^k_{U,r, A}\}|$.
\end{definition}

\subsection{Relation of different Definitions}\label{relations}

\subsubsection{Individual Unfairness v.s. Favored Candidate}

If $l$ candidates from group $G_1$ are treated individually unfair, then $l$ candidate from a group $G_2$ will be favored as the total number of selected candidates are $k$. 

\begin{proposition}
In top-k selected candidates, if a candidate $c_i$ from group $G_i$ is treated unfairly as defined in Definition \ref{def2}, then a candidate from any other group will be favored as defined in Definition \ref{def3}. 
\end{proposition}

\begin{proof}
We prove it using proof of contradiction. Let's assume there are two groups and a candidate $c_i$ from group $G_1$ is treated unfairly and no candidate of group $G_2$ is favored. In top-k elements, URR will return $round(p_{a_1}*k)$ candidates from group $G_1$ and $round(p_{a_2}*k)$ candidates from group $G_2$, where $round(p_{a_1}*k) + round(p_{a_2}*k) =k$. 

Now, if the candidate $c_i$ Group $G_1$ is treated unfairly then $c_i \notin \tau^k_{L,r, A}$ and $\tau^k_{L,r, A} (c_i) > k$, so the place of $c_i$ candidate in top-$k$ candidates will be vacant. Therefore, the platform $\mathbb{L}$ will return $round(p_{a_1}*k)-1$ candidates from group $G_1$ and $round(p_{a_2}*k)$ candidates from group $G_2$ as no candidate from group $G_2$ is favored. 
Now, the total number of returned candidates by platform $\mathbb{L}$ is $round(p_{a_1}*k)-1 + round(p_{a_2}*k) = k-1$. That is a contradiction as top-$k$ elements were selected by $L$. Therefore, by the proof of contradiction, an element from group $G_2$ will be favored.
\end{proof}

\subsubsection{Favored Candidate v.s. Individual Unfairness}

Intuitively it seems, If $l$ candidates from group $G_2$ will be favored, then $l$ candidates from group $G_1$ are treated individually unfairly as the total number of selected candidates are $k$. However, this is \textbf{not always true.}

\begin{proposition}\label{prop2}
In top-k selected candidates, if a candidate from group $G_i$ is favored as defined in Definition \ref{def3} then it does not imply that a candidate from another group will be treated unfairly as defined in Definition \ref{def2}.
\end{proposition}

\begin{proof}
We prove it using the proof of construction. Let's assume, in the example (Figure 1), only $g_2$ and $g_4$ joins the platform $\mathbb{L}$. Now $(G_1:G_2 =10:2)$, so, the new representative ranking in set $L$ is: $b_1 > b_2 > b_3 > g_2 > b_4 > b_5 > b_6 > b_7 > b_8 > g_4 > b_9 > b_{10}$. If top-6 candidates are selected,
\begin{itemize}
    \item $U$ will return: $b_1 > g_1 > b_2 > b_3 > g_2 > b_4$
    \item $L$ will return: $b_1 > b_2 > b_3 > g_2 > b_4 > b_5$
\end{itemize}
As you can see that the candidate $b_5$ is favored, but no candidate of group $G_2$ is treated unfairly. Hence proved.
\end{proof}

\subsubsection{Individual Unfairness v.s. Group Unfairness}

\begin{lemma}
If there is individual unfairness (as defined in Definition \ref{def2}) for a candidate $c_i$ and $c_i \in G_i$, it implies the group unfairness (as defined in Definition \ref{def5}) for group $G_i$.
\end{lemma}

This can be concluded using Definition \ref{def5}.

\subsubsection{Favored Group v.s. Group Unfairness}


\begin{lemma}
If there exists a group $G_i$ that is favored group, then it implies that there will also exist a group $G_j$ that suffers through group unfairness.
\end{lemma}

\begin{proof}
This will be proved by contradiction. Let's assume that no group is treated unfairly on platform $\mathbb{L}$ given that the group $G_i$ is favored. There are $n$ groups and in top-$k$ elements, the number of selected candidates from each group are $(l^U_{G_1}, l^U_{G_2}, ..., l^U_{G_n})$ using URR.

If group $G_i$ is favored on platform $\mathbb{L}$, then it implies that there is a candidate $c_i$ from group $G_i$ that is favored. Therefore, the number of selected candidates from group $G_i$ on platform $\mathbb{L}$ is $l^U_{G_i} + 1$. If no group is treated unfairly on $\mathbb{L}$, the total number of selected candidates on platform $\mathbb{L}$ is $(l^U_{G_1}, l^U_{G_2} + ... + l^U_{G_i} + 1+ \cdots + l^U_{G_n}) =k+1$. This is a contradiction to the condition that top-$k$ candidates were chosen. Hence, by contradiction, a group $G_j$ would have been treated unfairly.
\end{proof}

\subsubsection{\textbf{Individual Fairness v.s. Group Fairness}}

This is important to understand how individual and group fairness are correlated so that the appropriate solutions can be designed. A solution designed for individual fairness may not provide group fairness. However, a solution designed for group fairness leads to individual fairness. We will prove these two lemmas next. 

\begin{lemma}
The individual fairness-aware solution does not imply group fairness. 
\end{lemma}

\begin{proof}
Let's assume, in top-$k$ candidates, the number of chosen candidates by URR from groups $(G_1, \cdots, G_i, \cdots, G_2)$ are \\
$(l^U_{G_1}, \cdots, l^U{G_i}, \cdots, l^U_{G_2})$, respectively, and the candidates in set $H$ ($H=\{c_i | c_i \in G_i \; and \; c_i \in \tau^k_{U,r,A}\}$) does not join the platform $\mathbb{L}$. 

In this case, even if no candidate of $G_i$ is chosen by an algorithm in top-$k$ candidates, the algorithm is still providing an individual-fair solution as there is no individual unfairness as per Definition \ref{def2}; however, it is not a group-fair solution, as $l^U_{G_i}$ candidates from group $G_i$ should have been chosen in top-$k$ for adhering the group fairness as per the Definition \ref{def5}. Hence proved.
\end{proof}

\begin{lemma}
Group fairness-aware solution will imply individual fairness.
\end{lemma}

\begin{proof}
In a group fairness-aware solution $R$, for a candidate $c_i$, if $\tau{U,r,A}(c_i) \leq k$ then from the definition of group fairness $R(c_i) \leq k$ as the ratio of candidates for each group is maintained in group fairness ranking while maintaining the ranking of the candidates in their respective groups. Now, as per the definition of individual fairness (refer Definition \ref{def2}), a solution is individual fairness aware, if $\tau_{U,r,A}(c_i) \leq k$ then $R(c_i) \leq k$, that is already true. Hence proved.
\end{proof}

\section{Fair Representative Ranking}

In this section, we discuss what is an ideal solution for fair representative ranking on a platform $\mathbb{L}$ if the data is missing. In the following subsections, we propose methods to generate \textit{individual fair representative ranking (IFRR)} and \textit{group fair representative ranking (GFRR)} that adhere to the definitions \ref{def2} and \ref{def5}, respectively. The proposed ideal solution assumes that the information of universal data and the data of platform $\mathbb{L}$ is available. In the next section, we will propose solutions to generate fair representative rankings in the case of missing data that will compete with the Ideal fair representative rankings. 

\subsection{Individual Fair Representative Ranking}
A ranking is fair for an individual $c_i$ if  $c_i \in \tau_{U,r, A}^k$ then $c_i \in \tau_{L,r, A}^k$. In short, there should not be any candidate $c_i$ in the generated individual fair representative ranking $IFRR(c_i)$ such that the candidate satisfy the condition of individual unfairness $IUF(c_i)=1$ as defined in Definition \ref{def2}.

IFRR for platform $\mathbb{L}$ can be generated by projecting URR on the candidates of set $L$. It is denoted by $IFRR^L$, and the algorithm is given in Algo \ref{algo-ifrr}. In the algorithm, $range(1,n)$ function iterates from $1$ to $n$ and $len(\tau_{U,r, A})$ function returns the length of the given list. In steps 2 to 4, we iterate over each candidate $c_i$ of the URR, and if the candidate belongs to set $L$ ($c_i \in L$), it is added to $IFRR^L$.

\begin{algorithm}[t]
\caption{$Ideal-IFRR(\tau_{U,r, A},L)$} 
\label{algo-ifrr}
\SetKwInOut{Input}{Input}\SetKwInOut{Output}{Output}
\DontPrintSemicolon
  \Input{$\tau_{U,r, A}$: universal representative ranking, $L$: set of users on platform $\mathbb{L}$} 
  
  \Output{$IFRR^L$: Individual fair representative ranking on $\mathbb{L}$}
  $IFRR^L$=[ ] \;
    \For{$k$ in range($1$, len($\tau_{U,r, A}$))}    
    { 
    	  \If{$\tau_{U,r, A}[k] \in L$}
            {
                append $\tau_{U,r, A}[k]$ to $IFRR^L$
            }
    }
    return $IFRR^L$
\end{algorithm}


\subsection{Group Fair Representative Ranking}

A ranking solution is group fair representative ranking (GFRR) if no group $G_i$ satisfies the condition of group unfairness as defined in Definition \ref{def5}. An ideal-GFRR solution ensures that the number of candidates from a group $G_i$ in GFRR should be equal to the number of candidates from the group $G_i$ in URR, except if all the members of this group are already ranked in GFRR. Therefore, for each index $k$, an ideal-GFRR solution will maintain the group fairness for each group $G_i$ until all the elements of that group are exhausted. Therefore, a ranking is fair for a group $G_i$ having protected value $a_i in A$, if $|{c_i \; s.t. \; c_i \in \tau_{L,r, A}^k and A(c_i)=a_i}|$=\\$min(|{c_i \; s.t. \; c_i \in \tau_{U,r, A}^k \; and \; A(c_i)=a_i}|$, $|{c_i \; s.t. \; c_i \in L \; and \; A(c_i)=a_i}|)$.

It is denoted by $GFRR^L$, and the method is explained in Algorithm  \ref{algo-gfrr}. For each index $k$ of URR, if the candidate $\tau_{U,r, A}[k]$ is in set $L$, it is added to $GFRR^L$; otherwise, the next eligible candidate of this protected group is added to $GFRR^L$. The $getAttribute(c)$ function returns the protected attribute of candidate $c$, $getNext(R,a)$ function returns the first element of ranked list $R$ having protected attribute $a$, if there is no remaining element having this attribute, it returns `NA'. In Algorithm \ref{algo-gfrr}, step 4 generates Ideal-IFRR that will only contain the candidates belonging to set $L$, as it will help in quickly run the $getNext$ function to identify the next ranked element of the given protected attribute. In step 3 to 8: step 3 iterates over each candidate of URR, step 4 identifies the protected attribute of the iterated candidate, step 5 gets the next eligible candidate of this protected group, step 6 will verify if a candidate is returned, step 7 will remove this candidate from $IFRR^L$, and step 8 will add it to $GFRR^L$.

\begin{algorithm}[t]
\caption{$Ideal-GFRR(\tau_{U,r, A},L,A)$}
\label{algo-gfrr}
\SetKwInOut{Input}{Input}\SetKwInOut{Output}{Output}
\DontPrintSemicolon
  \Input{$\tau_{U, r, A}$: universal representative ranking, $L$: set of users on platform $\mathbb{L}$, $A$: set of protected attribute}
  
  \Output{$GFRR^L$ Group fair representative ranking on $\mathbb{L}$}
  $GFRR^L$=[ ] \;
  $IFRR^L$=$Ideal-IFRR^L(\tau_{U,r, A},L)$ \;
    \For{$k$ in range(1, len($\tau_{U,r, A}$))}    
    { 
        $a=getAttribute(\tau_{U,r, A}[k])$\;
        $c$= $getNext(IFRR^L,a)$\;
        \If{$c$ != `NA'}
        {
            remove $c$ from $IFRR^L$ \;
            append $c$ to $GFRR^L$ \;
        }
    }
    return $GFRR^L$ \;
\end{algorithm}


\section{Problem Statement}

By far, we have discussed how the representative ranking is unfair on a platform $\mathbb{L}$ if some of the candidates do not join $\mathbb{L}$; the main reason for the unfairness is that the information of universal data is missing. If the universal representation ratio is known, then the ideal IFRR and GFRR can be generated using Algorithm 1 and 2. However, in most real-world applications, the information of the universal dataset is unknown. For example, for a given skill set $s$, how many males and females candidates are eligible for the job, that shows the universal representation ratio, might not be known to any given platform $\mathbb{L}$.

\textbf{Problem Statement.} \textit{Given set $L$ containing eligible candidates for a search query $r$, protected attribute set $A$, and the number of selected candidates $k$, propose methods to generate IFRR and GFRR.}


The proposed solutions should maintain IFRR or GFRR based on the representation ratio of different classes of protected attributes. Therefore, the solutions are categorized into two following classes.

\begin{enumerate}
    \item If universal representation ratio is known: In this work, we will propose solutions and possible approaches to generate optimal IFRR and GFRR solutions if the universal representation ratio is known. 
    \item If universal representation ratio is unknown: For such cases, we first need to estimate the universal representation ratio, and then the solutions proposed for the first category can be applied. The solutions to estimate representation ratio can be designed using statistics-based techniques or using surveys. This is an open question and not covered in this paper.
\end{enumerate}

\section{The Proposed Solutions}\label{solution}

In this section, first, we discuss a method to generate an optimal solution for GFRR. The solution for generating optimal IFRR is not straightforward as IFRR ranking is highly dependent on the kind of missing data; we will discuss these complexities and possible IFRR solutions in the next subsection. 

\subsection{Generate GFRR}

If the universal representation ratio is known, then a method to generate GFRR will maintain the ratio of candidates based on protected attribute classes for each index $k$. The method is inspired from \cite{geyik2019fairness}, and the solution is explained in Algorithm \ref{generategfrr}. The generated GFRR ranking is represented by $\widetilde{GFRR^L}$. In the algorithm, function $score(a_i, counts[a_i])$ returns the score of next eligible candidate from protected class $a_i$ (or the score of $counts[a_i]$ ranked candidate from the group $G_i$), $getNextCand(a_i)$ function returns the next eligible candidate of group $G_i$ having protected attribute $a_i$ (or $counts[a_i]$ ranked candidate from the group $G_i$).

The generated solution will be the same as an Ideal-GFRR; as both $\widetilde{GFRR^L}$ and ${GFRR^L}$ rankings maintain the representation ratio for each index $k$.

\begin{algorithm}[t]
\caption{Generate-$\widetilde{GFRR^L}$}
\label{generategfrr}
\SetKwInOut{Input}{Input}\SetKwInOut{Output}{Output}
\DontPrintSemicolon
  \Input{Protected attribute $A=\{a_1, a_2, ..., a_l\}$, universal representation ratio $\{p_{a_1}, p_{a_2}, ...,  p_{a_l}\}$, Elements for each class of protected attribute: $\{n_{a_1}, n_{a_2}, ...,  n_{a_l}\}$} 
  
  \Output{Representative ranking $\tau$ for group fairness}
  $\widetilde{GFRR^L}$ =[] \; 
  $counts=\{\}$ \\
  \For{each $a_i$ in A}    
    { 
    	$counts[a_i]=0$
    }
    $N=\sum_{a_i}n_{a_i}$ \\
    \For{k in range(1,N)}    
    { 
    	$underRep=\{a_i | counts[a_i]<\left \lfloor k \cdot p_{a_i}\right \rfloor \; and \; counts[a_i] < n_{a_i}\}$ \;
    	$underLimit=\{a_i | counts[a_i] \geq \left \lfloor k \cdot p_{a_i}\right \rfloor \; and \; counts[a_i] < \left \lceil k \cdot p_{a_i}\right \rceil \; and \; counts[a_i] < n_{a_i}\}$ \;
    	  \If{$underRep \neq \phi$}
            {
                $selectAttr=argmax_{a_i \in underRep}score(a_i, counts[a_i])$ \;
            }
            \Else
            {
            	$selectAttr=argmin_{a_i \in underLimit}\frac{\left \lceil k \cdot p_{a_i}\right \rceil}{p_{a_i}} $\;
            }
        $selectCand=getNextCand(selectAttr)  $ \;
        $\widetilde{GFRR^L}[k]=selectCand$ \;
        $counts[selectAttr]=counts[selectAttr]+1$\;
    }
    return $\widetilde{GFRR^L}$\;
\end{algorithm}

\subsection{Generate IFRR}

The solutions to generate IFRR are not easy, even if the universal representation ratio is known. We will explain it with the help of an example. Let's assume, in the universal set $U$ there are two groups $(G_1, G_2)$ and for search query $r$, in group $G_1$, there are four eligible candidates $(b_1, b_2, b_3, b_4)$ and in group $G_2$, there are four eligible candidates $(g_1, g_2, g_3, g_4)$. We assume that in each group the rank of $c_i > c_j$ if $i<j$. So, IFRR on set $U$ is $(g_1, b_1, g_2, b_2, g_3, b_3, g_4, b_4)$ if group $G_2$ is prioritized. Let's assume that four candidates from group $G_1$ and two candidates from group $G_2$ joins the platform $\mathbb{L}$, therefore $L= \{b_1, b_2, b_3, b_4, g'_1, g'_2\}$. Now, the question is: if the universal representation ratio is known as (0.5:0.5), how to generate an individual fair representative ranking. It is not straight forward as we do not know which two out of four members of group $G_2$ joined $\mathbb{L}$. Let's further understand this with case scenarios. In case 1: If $g'_1$ is $g_1$ and $g'_2$ is $g_2$, then the Ideal-IFRR will be $(g'_1, b_1, g'_2, b_2, b_3, b_4)$. In case 2: If $g'_1$ is $g_1$ and $g'_2$ is $g_4$, then the Ideal-IFRR will be $(g'_1, b_1, b_2, b_3, g'_2, b_4)$. Therefore, an IFRR solution also require the knowledge of missing data or the knowledge of what kind of candidates have joined $\mathbb{L}$ (metadata for platform $\mathbb{L}$).

The information of what kind of candidates have joined the platform $\mathbb{L}$ can either be gathered by performing surveys or by using statistics-based methods. The surveys can help in knowing what is the probability of a candidate joining $\mathbb{L}$ and what is its correlation with candidate's characteristics. For example, if we know that the higher ranked people are more likely to join a platform $\mathbb{L}$, then an estimated IFRR solution for the above discussed example is: $(g'_1, b_1, g'_2, b_2, b_3, b_4)$. In another approach, the statistical methods can be used to estimate the actual distribution using the partial available information.

In this work, we show the results for the case when the candidates from a given group $G_i$ are missing uniformly at random. In the proposed solution, for each ranking place reserved for a sub-active group, its net ranked candidate will be placed with the probability equal to the fraction of active people of this group. The method is represented by $\widetilde{IFRR^L}$. The proposed solution can be extended for the cases where the probability of joining a candidate is correlated with its URR, such as (i) the probability of joining is directly proportional to score, (ii) the probability of joining is inversely proportional to score, or (iii) the probability of joining is higher for middle level candidates. The platform specific solutions for other cases are open questions. 

\section{Simulation Results}

\subsection{Metrics}

For the purposes of evaluation, we have used the following metrics.

\begin{enumerate}
    \item \textbf{Rank Difference:} This measure is used to capture the individual unfairness for ranking $R_1$ v.s. ranking $R_2$. It is computed as,
    
    \begin{equation}
        RankDiffer(c_i)= R_1(c_i) - R_2(c_i)
    \end{equation}
    
    The positive value shows that the individual has been treated unfairly in ranking $R_1$ with respect to ranking $R_2$, and the value shows the difference in the ranking. The negative value shows that the candidate is favored on $R_1$. 
    
    \item \textbf{Skew Ratio:} Skew ratio denotes the extent to which the top-$k$ elements for a given search query and for a given sub-group having the given attribute value $a_i$ differ from the desired proportion of that attribute group. It is used to measure group-unfairness, and it is defined as, \\
    
    \begin{equation}
        Skew_{R_1,R_2}^{a_i}=log_{e}\left ( \frac{|\{c_i \; s.t. \; c_i \in R_1^k \; and \; A(c_i)=a_i\}|}{|\{c_i \; s.t. \; c_i \in R_2^k \; and \; A(c_i)=a_i\}|} \right )
    \end{equation}
    
    $Skew_{R_1,R_2}^{a_i}$ is the (logarithmic) ratio of the proportion of candidates having the attribute value $a_i$ among the top-k ranked results to the desired proportion for $a_i$ in a ranking. A negative value corresponds to a lesser than desired representation of candidates with value $a_i$ in the top-k results, while a positive value corresponds to favoring such candidates. The log is used to make the skew values symmetric around the origin with respect to ratios for and against a specific attribute value $a_i$. 
    
    Note that the given formula doesn't work for $log(0)$ when the representation is zero, and in this case, we can assign a large value to this metric.
    
    
\end{enumerate}

\subsection{Datasets}
The experimental results are shown on synthetic datasets so that the impact of different activeness can be better studied and various values of protected attributes can be considered. The data has a protected attribute having 3 values $A=\{g,b,u\}$ and so the candidates are categorized into three groups $G_g, G_b, G_u$. A real-world example of this situation is gender classes having `female', `male', and `unknown' values. For a search query $r$, in each group, there are 10000 eligible candidates, and their scores are assigned uniformly at random from (0,1]. We assume that group $G_g$ is less active on the platform $\mathbb{L}$, and the ratio of active members ranges from 0.1 to 0.9, and all candidates of $G_b$ and $G_u$ join $\mathbb{L}$.

\subsection{Discussion}


We perform the experiment to show the number of candidates from group $G_g$ who miss the opportunity on platform $\mathbb{L}$ as $k$ (the number of selected candidates using representative ranking LRR) and activeness of group $G_g$ varies. The results are shown in Figure \ref{plot1}; the plot shows that the number of candidates losing the opportunity increases linearly with both parameters.

\begin{figure}
    \centering
    \includegraphics[width=\linewidth]{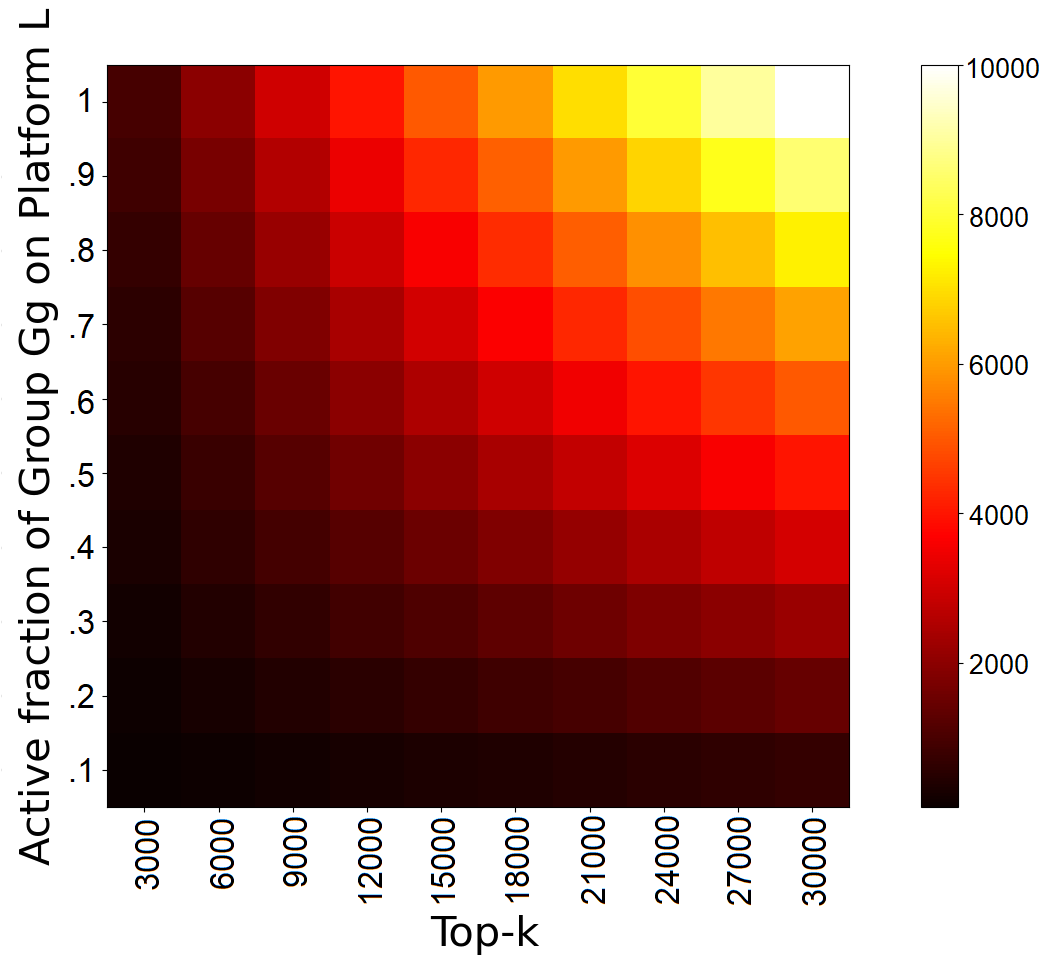}
    \caption{The number of candidates from group $G_g$ who miss the opportunity in top-$k$ selected candidates as $k$ and the percentage of active people vary.}
    \label{plot1}
\end{figure}

Figure \ref{rankdiff} shows the rank-difference of LRR and the proposed IFRR with Ideal-individual fairness representative ranking. The results show that the performance of LRR is poor for the activeness ratio from 0.4 if the candidates join u.a.r. However, detailed solutions will be proposed further. Figure \ref{skew} shows the skewness of LRR versus ideal-group fairness representative ranking, and the proposed method performs as well as the ideal method as both maintain the representation ratio. 

\begin{figure}[t]
    \centering
    \includegraphics[width=\linewidth]{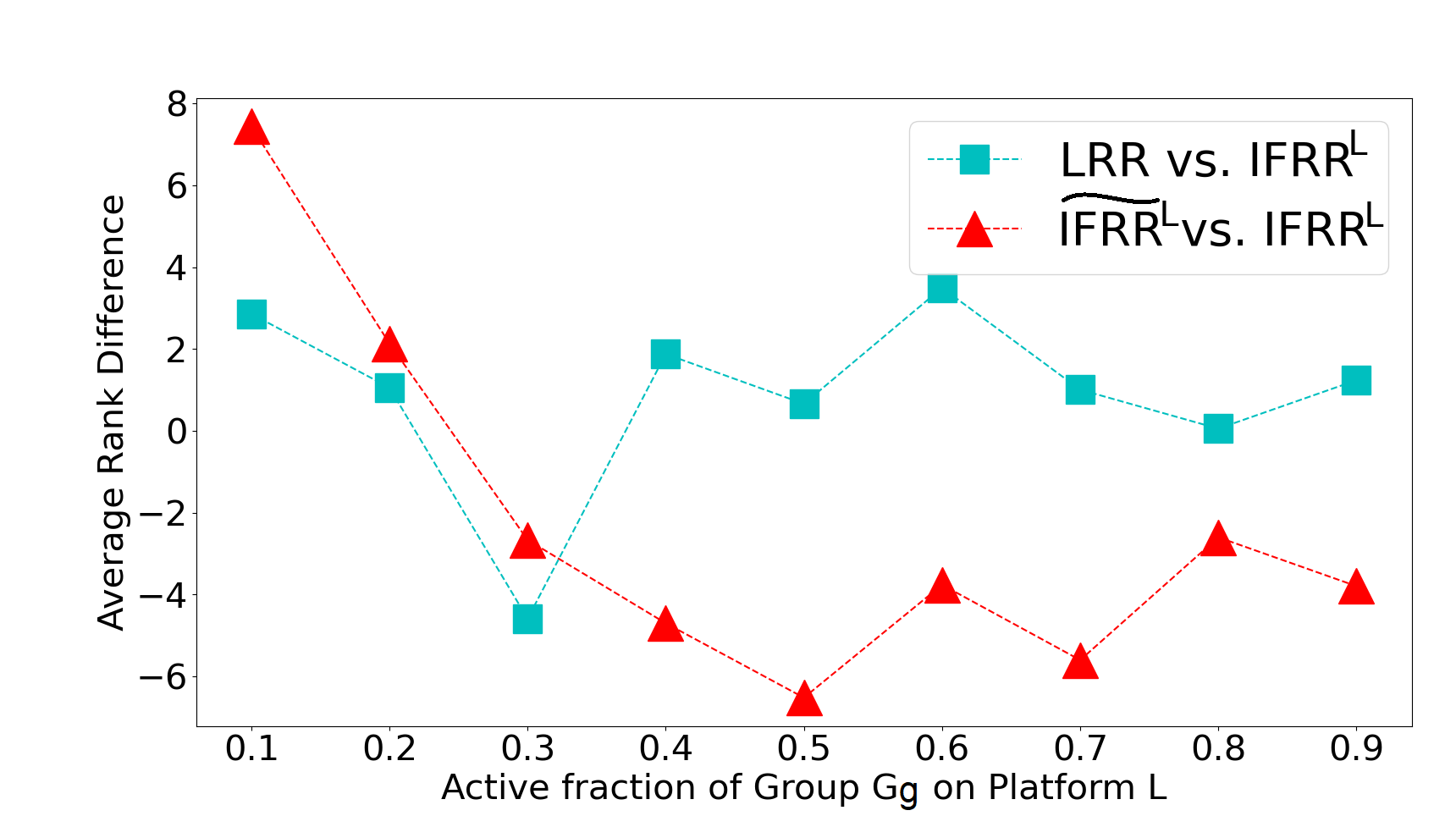}
    \caption{Average Rank Difference versus active fraction of group $G_g$ for group $G_g$.}
    \label{rankdiff}
\end{figure}

\begin{figure}[t]
    \centering
    \includegraphics[width=\linewidth]{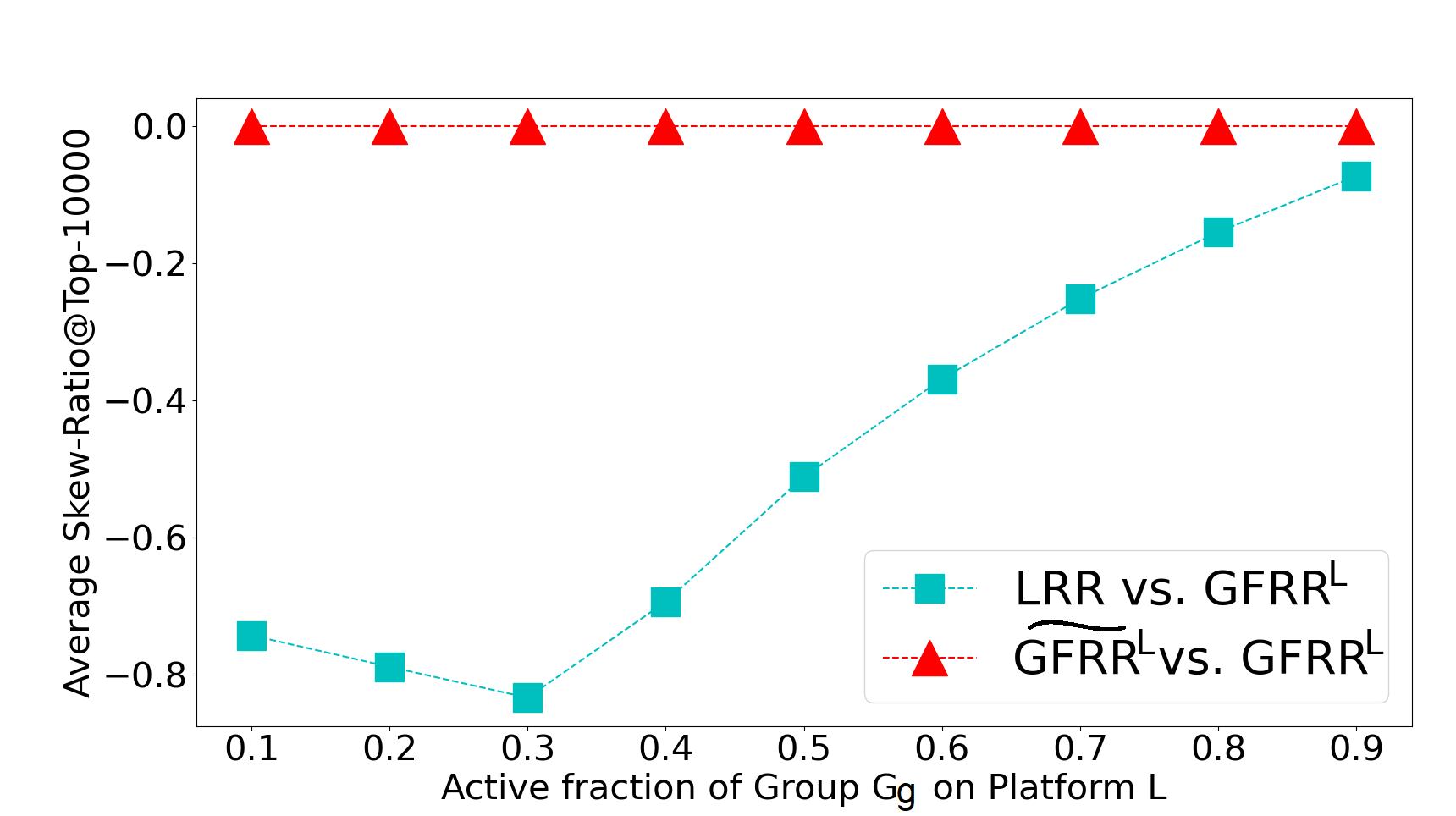}
    \caption{Average Skew ratio versus active fraction of group $G_g$ for top-10000 ranking for group $G_g$.}
    \label{skew}
\end{figure}



\section{Related Work}\label{related}


In the past two decades, several works have focused on defining and quantifying the extent of discrimination \cite{pedreshi2008discrimination, caliskan2017semantics, prates2019assessing}. Several strategies have been proposed for bias mitigation from machine learning models \cite{castro2019s, bird2019fairness, lepri2018fair, celis2016fair,friedler2019comparative, hardt2016equality}. There has been extensive work on algorithmic bias and discrimination in solutions applied to various disciplines, such as computer science, law, policy \cite{friedman1996bias, hajian2016algorithmic, vzliobaite2017measuring}. 

The recent works have discussed the limitations of different notions of fairness and non-discrimination \cite{corbett2017algorithmic,dwork2012fairness,friedler2016possibility}. Researchers have investigated both notions of fairness, (i) individual fairness that focuses on that similar people should be treated similarly \cite{dwork2012fairness, binns2020apparent}, and (ii) group fairness which focuses on that each group, either advantaged or disadvantaged, should be treated similarly \cite{pedreshi2008discrimination, lohia2019bias, pedreschi2009measuring}. These two fairness definitions represent intrinsically different point of views, and accommodating both fairness in a solution requires trade-offs \cite{friedler2016possibility}. In this work, the people holding similar ranks in their respective groups are considered equally eligible, as also shown in previous works \cite{geyik2019fairness}; and IFRR focuses that they should be treated equally. However, GFRR follows the requirement of group fairness. 

Next, we discuss state-of-the-art literature on fairness-aware ranking methods as our work is closely related to them \cite{asudeh2019designing, biega2018equity, castillo2019fairness, celis2017ranking, zehlike2017fa, yang2017measuring, singh2018fairness}. Asudeh et al. \cite{asudeh2019designing} proposed a method for generating fair ranking where the user has the flexibility to choose a weight that should be assigned to each criterion, within limits. The authors developed a system that is helpful for users in choosing criterion weights so that the desired fairness can be achieved efficiently and effectively. Celis et al. \cite{celis2017ranking} presented a theoretical investigation of ranking with fairness and diversity constraints. The authors proposed fast exact and approximation methods for the constrained ranking maximization problem. The proposed method runs in linear time, even if there are a large number of constraints, and the approximation ratio does not depend on the number of constraints.

Singh and Joachims \cite{singh2018fairness} proposed a framework to generate fairness ranking for the given fairness constraints in terms of exposure allocation. The proposed framework presents efficient methods for generating rankings that maximize the utility for the user while satisfying the given notion of fairness. The authors showed how a large range of fairness constraints could be incorporated using the proposed framework, such as demographic parity, disparate treatment, and disparate impact constraints. Beiga et al. \cite{biega2018equity} studied position bias where the low-ranked subjects achieve less attention that shows individual unfairness. The authors proposed a linear programming based solution to achieve amortized fairness where attention accumulated across a series of rankings is proportional to accumulated relevance.



Kulshrestha et al. \cite{kulshrestha2017quantifying} studied search bias in rankings while returning results for a search query on social media and quantified to what extent this output bias is due to the dataset that was used as the input to the ranking system and what extent is introduced due to the ranking system itself. They proposed a framework and verified it on Twitter for politics-related queries. The authors showed that both the input data as well as the ranking system contribute significantly to introduce varying amounts of bias in the search results. They further discussed the impact of such biases and possible methods to notify the user about the existing bias. The authors mentioned that it would be much more helpful if the users are given control to re-rank the search results in the presence of bias than just notify them about its existence.

Yang and Stoyanovich \cite{yang2017measuring} studied fairness in ranking and proposed fairness measures for ranked outputs to compare the distributions of protected and non-protected candidates on different prefixes of the ranked list. The authors also showed preliminary results of incorporating their proposed fairness measures into an optimization framework for improving the fairness of ranked outputs while maintaining accuracy. Zehlike et al. \cite{zehlike2017fa} studied fairness in top-$k$ selected candidates where the goal was to ensure that the proportion of candidates based on protected attributes remains statistically above a given minimum for every prefix of the ranking. This work was further extended by Geyik et al. \cite{geyik2019fairness}, and they verified the method on LinkedIn Talent Search and observed three times improvement. In our work, we have shown the unfairness in this kind of methods if the platform is not aware of the universal representation ratio.

\section{Conclusion and Further Direction}\label{conclusion}


In this paper, we discussed how a fairness-aware ranking method might be unfair due to the unavailability of universal information. We proposed individual and group fairness definitions that should be considered while verifying the fairness of the proposed representative ranking if some groups are less active. We also discussed the correlation and dependencies of the proposed definitions. We further proposed methods to generate ideal individual and group fair representative rankings that will work as a baseline. Next, we discussed methods and possible approaches to generate optimal IFRR and GFRR if the universal representation ratio is known. The experiments show preliminary results for unfairness increases in LRR as the group becomes less active and the improvements achieved by proposed methods. 

The main challenge in real-life applications is to identify the universal representation ratio based on the given protected attribute. In this work, we highlighted two approaches to solve this (i) perform a survey to understand the existing distribution, (ii) apply statistical methods to estimate the exact distribution based on the available score-ranking information of the candidates belonging to different protected attributes based groups. These solutions will be discussed in future work.

\bibliographystyle{ACM-Reference-Format}
\bibliography{linkedinbib}

\end{document}